\documentclass[titlepage,babel,USenglish]{cc} 
\pdfoutput=1

\usepackage{latexsym}
\usepackage{amssymb}
\usepackage{amsmath}
\usepackage{float}
\usepackage{comment}
\usepackage[svgnames]{xcolor}
\usepackage[pdfauthor={Mark Giesbrecht and Daniel S. Roche},
  pdftitle={Interpolation of Shifted-Lacunary Polynomials},
  pdfsubject={Symbolic Computation, Data Structures and Algorithms},
  pdfkeywords={Sparse interpolation, Sparsest shift, Lacunary polynomials},
  colorlinks,linkcolor=DarkBlue,citecolor=DarkGreen]{hyperref}

\DeclareMathAlphabet{\mathbold}{OML}{cmm}{b}{it}

\makeatletter
\long\def\unmarkedfootnote#1{{\long\def\@makefntext##1{##1}\footnotetext{#1}}}
\makeatother

\renewcommand{\thefootnote}{\fnsymbol{footnote}}

\newcommand{\ZZ}{{\mathbb{Z}}}
\newcommand{\QQ}{{\mathbb{Q}}}
\newcommand{\NN}{{\mathbb{N}}}
\newcommand{\CC}{{\mathbb{C}}}
\renewcommand{\P}{{\mathcal{P}}}
\newcommand{\Q}{{\mathcal Q}}
\newcommand{\M}{{\mathsf{M}}}

\newcommand{\F}{{\mathsf{F}}}
\newcommand{\fp}{f^{(p)}}
\newcommand{\bs}{B_\Sigma}
\newcommand{\br}{B_S}
\newcommand{\bb}{\kappa_f}
\newcommand{\ndivs}{\nmid}
\newcommand{\divs}{{\mskip3mu|\mskip3mu}}

\newcommand{\floor}[1]{{\lfloor #1 \rfloor}}
\newcommand{\pset}{\mathcal{G}}

\DeclareMathOperator{\lcm}{lcm}
\DeclareMathOperator{\nequiv}{\mskip4mu\not\equiv\mskip4mu}
\DeclareMathOperator{\size}{size}
\newcommand{\sz}[1]{\size(#1)}
\DeclareMathOperator{\rem}{rem}
\DeclareMathOperator{\remo}{rem_1}
\DeclareMathOperator{\llog}{loglog}
\DeclareMathOperator{\lllog}{logloglog}
\renewcommand{\th}{{\rm th}}

\contact{mwg@cs.uwaterloo.ca}
\received{\today}
\title{Interpolation of Shifted-Lacunary Polynomials}

\author{
  Mark Giesbrecht\\
  School of Computer Science\\
  University of Waterloo\\
  Waterloo, ON, N2L 3G1, Canada\\
  \email{mwg@cs.uwaterloo.ca}\\
  \homepage{http://www.cs.uwaterloo.ca/~mwg}
\and
  Daniel S. Roche\\
  School of Computer Science\\
  University of Waterloo\\
  Waterloo, ON, N2L 3G1, Canada\\
  \email{droche@cs.uwaterloo.ca}\\
  \homepage{http://www.cs.uwaterloo.ca/~droche}
}

\begin{abstract}
  Given a ``black box'' function to evaluate an unknown rational
  polynomial $f\in\QQ[x]$ at points modulo a prime $p$, we exhibit
  algorithms to compute the representation of the polynomial in the
  sparsest shifted power basis.  That is, we determine the sparsity
  $t\in\ZZ_{>0}$, the shift $\alpha\in\QQ$, the exponents
  $0\leq e_1<e_2<\cdots <e_t$, and the coefficients
  $c_1,\ldots,c_t\in\QQ\setminus\{0\}$ such that
  \[
  f(x) = c_1(x-\alpha)^{e_1}+c_2(x-\alpha)^{e_2}+\cdots+c_t(x-\alpha)^{e_t}.
  \]
  The computed sparsity $t$ is \emph{absolutely} minimal over any
  shifted power basis.  The novelty of our algorithm is that the
  complexity is polynomial in the (sparse) representation size,
  which may be logarithmic in the degree of $f$.
  Our method combines
  previous celebrated results on sparse interpolation and computing
  sparsest shifts, and provides a way to handle polynomials with
  extremely high degree which are, in some sense, sparse in information.
\end{abstract}


\begin{subject}
 Primary 68W30; Secondary 12Y05
\end{subject}

\begin{document}

\maketitle

\section{Introduction}
Interpolating an unknown polynomial from a set of evaluations is a problem
which has interested mathematicians for hundreds of years, and which is now
implemented as a standard function in most computer algebra systems.
To illustrate some different kinds of interpolation problems, consider the
following three representations for a polynomial $f$ of degree $n$:
\begin{eqnarray} 
\label{eqn:denserep}
f(x) &=& a_0 + a_1x + a_2 x^2 + \cdots + a_n x^n, \\
\label{eqn:sparserep}
f(x) &=& b_0 + b_1 x^{d_1} + b_2 x^{d_2} + \cdots + b_s x^{d_s}, \\
\label{eqn:ssrep}
f(x) &=& c_0 + c_1(x-\alpha)^{e_1}+c_2(x-\alpha)^{e_2}+\cdots+c_t(x-\alpha)^{e_t}.
\end{eqnarray}
In \ref{eqn:denserep} we see the dense representation of the
polynomial, where all coefficients (even zeroes) are
represented. Newton and Waring discovered methods to interpolate $f$
in time proportional to the size of this representation in the
$18^\th$ century.  The sparse or lacunary representation is shown in
\ref{eqn:sparserep}, wherein only the terms with non-zero coefficients
are written (with the possible exception of the constant coefficient
$b_0$).  Here we say that $f$ is \emph{$s$-sparse} because it has
exactly $s$ non-zero and non-constant terms; the constant coefficient
requires special treatment in our algorithms regardless of whether or
not it is zero, and so we do not count it towards the total number of
terms.  \citet*{benor-tiwari} discovered a method to interpolate in
time polynomial in the size of this representation.
\citet{kal-lee-early} present and analyze very efficient algorithms
for this problem, in theory and practice.  The
\citeauthor{benor-tiwari} method has also been examined in the context
of approximate (floating point) polynomials by \citet{gll-interp},
where the similarity to the 1795 method of de Prony is also pointed
out.  \cite{BlaHar09} consider the more basic problem of identity testing of
sparse polynomials over $\QQ$, and present a deterministic
polynomial-time algorithm.

In \ref{eqn:ssrep}, $f$ is written in the shifted power basis 
$1,(x-\alpha),(x-\alpha)^2,\ldots$, and we say that $\alpha$ is a 
\emph{$t$-sparse shift} of $f$ because the representation has exactly
$t$ non-zero and non-constant terms in this basis.
When $\alpha$ is chosen so that $t$ is absolutely minimal in \ref{eqn:ssrep},
we call this the \emph{sparsest shift} of $f$. 
We present new algorithms
to interpolate $f \in \QQ[x]$, given a black box for evaluation, in time
proportional to the size of the shifted-lacunary representation corresponding
to \ref{eqn:ssrep}. It is easy to see that $t$ could be exponentially smaller
than both $n$ and $s$, for example when $f = (x+1)^n$, demonstrating
that our algorithms are providing a significant improvement in
complexity over those previously known, whose running times are polynomial 
in $n$ and $s$.

The main applications of all these methods for polynomial interpolation are
signal processing and reducing intermediate expression swell. Dense and sparse
interpolation have been applied successfully to both these ends, and our new
algorithms effectively extend the class of polynomials for which such
applications can be made.

The most significant challenge here is computing the sparsest shift
$\alpha\in\QQ$. Computing this value from a set of evaluation
points was stated as an open problem by \citet{borodin-tiwari}. An
algorithm for a generalization of our problem in the dense
representation was given by \citet{grig-karp-ssi}, though its cost is
exponential in the size of the output; they admit that
the dependency on the degree of the polynomial is probably not
optimal. Our algorithm achieves deterministic polynomial-time
complexity for polynomials over the rational numbers.  We are always
careful to count the \emph{bit complexity} --- the number of
fixed-precision machine operations --- and hence account for any
coefficient growth in the solution or intermediate expressions.

The black box model we use is slightly modified from the traditional one:
\bigskip

\setlength{\unitlength}{1 cm}
\begin{center}
\begin{picture}(8,1)(-.3,-.3)
\thicklines
\put(-.3,.4){$p \in \mathbb{N},\theta \in \ZZ_p$}
\put(2.2,.5){\vector(1,0){.8}}
\put(5,.5){\vector(1,0){.8}}
\put(6,.4){$ f(\theta) \bmod p$}
\put(2.95,-.3){$f(x) \in \QQ[x]$}
\linethickness{.8\unitlength}
\put(3,.5){\line(1,0){2}}
\end{picture}
\end{center}

\noindent%
Given a prime $p$ and an element $\theta$ in $\ZZ_p$, the black box
computes the value of the unknown polynomial evaluated at $\theta$
over the field $\mathbb{Z}_p$. (An error is produced exactly in those
unfortunate circumstances that
$p$ divides the denominator of $f(\theta)$.)
We generally refer to this as a \emph{modular black box}. 
To account for the reasonable possibility that the cost of black box
calls depends on the size of $p$, we define $\bb$ to be an upper bound
on the number of field operations in $\ZZ_p$ used in black box
evaluation, for a given polynomial $f\in\QQ[x]$.

Some kind
of extension to the standard black box, such as the modular black box
proposed here, is in fact necessary, since the value of a polynomial
of degree $n$ at any point other than $0,\pm 1$ will typically have
$n$ bits or more.  Thus, any algorithm whose complexity is
proportional to $\log n$ cannot perform such an evaluation over $\QQ$
or $\ZZ$.  Other possibilities might include allowing for evaluations
on the unit circle in some representation of a subfield of $\CC$, or
returning only a limited number of bits of precision for an
evaluation.

To be precise about our notion of size, first define 
$\sz{q}$ for $q \in \QQ$ to be the number of bits needed to represent $q$.
So if we write $q=\frac{a}{b}$ with $a \in \ZZ$, $b \in \NN$, and 
$\gcd(a,b)=1$, then $\sz{q}=\lceil \log_2 (|a|+1) \rceil +
\lceil \log_2 (b+1) \rceil + 1$.
For a rational
polynomial $f$ as in \ref{eqn:ssrep}, define:
\begin{equation} \label{eqn:sssize}
\sz{f} = {\rm size}(\alpha) + 
\sum_{i=0}^t \sz{c_i} +
\sum_{i=1}^t \sz{e_i}.
\end{equation}
We will often employ the following upper bound for simplicity:
\begin{equation}\label{eqn:ssize}
\sz{f} \leq {\rm size}(\alpha) +
t\left( H(f) + \log_2 n \right),
\end{equation}
where $H(f)$ is defined as  $\max_{0\leq i \leq t} \sz{c_i}$.

Our algorithms will have polynomial complexity in the smallest
possible $\sz{f}$.
For the complexity analysis, we use the normal notion of a
``multiplication time'' function $\M(n)$, which is the number of field
operations required to compute the product of polynomials with degrees
less than $n$, or integers with sizes at most $n$. 
We always assume that $\M(n) \in \Omega(n)$ and $\M(n) \in O(n^2)$. 
Using the 
results from \citet{CanKal91}, 
we can write $\M(n) \in O(n\log n \log\log n)$.

The remainder of the paper is structured as follows.  
In \ref{sec:sshift} we
show how to find the sparsest shift from evaluation points in $\ZZ_p$,
where $p$ is a prime with some special properties provided by some
``oracle''.  In \ref{sec:interp} we show how to perform sparse interpolation
given a modular black box for a polynomial.  In \ref{sec:primes} we show how
to generate primes such that a sufficient number satisfy the
conditions of our oracle.  \ref{sec:complexity} provides the complexity
analysis of our algorithms. We conclude in \ref{sec:conc}, 
and introduce some open questions.

\section{Computing the Sparsest Shift}
\label{sec:sshift}

For a polynomial $f \in \QQ[x]$,
we first focus on computing the sparsest shift $\alpha\in\QQ$ 
so that $f(x+\alpha)$
has a minimal number of non-zero and non-constant terms.
This information will later be
used to recover a representation of the unknown polynomial.

\subsection{The polynomial  $\mathbold \fp$}

Here, and for the remainder of this paper, for a prime $p$ and $f \in
\QQ[x]$, define $\fp\in\ZZ_p[x]$ to be the unique polynomial with degree less
than $p$ which is equivalent to $f$ modulo $x^p-x$ and with all
coefficients reduced modulo $p$. From Fermat's Little Theorem, we then
see immediately that $\fp(\alpha) \equiv f(\alpha) \bmod p$ for all
$\alpha \in \ZZ_p$. Hence $\fp$ can be found by evaluating $f$ at each
point $0,1,\ldots,p-1$ modulo $p$ and using dense interpolation over
$\ZZ_p[x]$.

Notice that, over $\ZZ_p[x]$, 
$(x-\alpha)^p \equiv x-\alpha \mod x^p-x$,
and therefore $(x-\alpha)^{e_i} \equiv (x-\alpha)^k$ for any $k \neq 0$
such that $e_i \equiv k \bmod (p-1)$. The smallest such $k$ is in the
range $\{1,2,\ldots,p\}$; we now define this with some more notation.
For $a \in \ZZ$
and positive integer $m$, define $a \remo m$ to be the unique integer
in the range $\{1,2,\ldots,m\}$ which is congruent to $a$ modulo
$m$. As usual, $a \rem m$ denotes the unique congruent integer
in the range $\{0,1,\ldots,m-1\}$.

If $f$ is as in \ref{eqn:ssrep}, then by reducing term-by-term we can
write 
\begin{equation}\label{eqn:fp}
\fp(x) = (c_0 \rem p) + \sum_{i=1}^t (c_i \rem p)(x-\alpha_p)^{e_i \remo (p-1)},
\end{equation}
where $\alpha_p$ is defined as $\alpha \rem p$. Hence, for some 
$k\leq t$,
$\alpha_p$ is a $k$-sparse shift for $\fp$.  That is,
the polynomial $\fp(x+\alpha_p)$ over $\ZZ_p[x]$ has at most $t$ non-zero
and non-constant terms.

Computing $\fp$ from a modular black box for $f$ is straightforward.
First, use $p$ black-box calls to determine $f(i)\rem p$ for
$i=0,1,\ldots,p-1$. Recalling that $\bb$ is the number of field
operations in $\ZZ_p$ for each black-box call, the cost of this step is
$O(p\bb \M(\log p))$ bit operations. Second, we use the well-known
divide-and-conquer method to interpolate $\fp$ into the dense
representation (see, e.g., \citet[Section 4.5]{BorMun75}). 
Since $\deg \fp < p$, this step has bit complexity 
$O(\M(p)\M(\log p)\log p)$.

Furthermore, for any $\alpha\in\ZZ_p$, the dense representation of
$\fp(x+\alpha)$ can be computed in exactly the same way as the second
step above, simply by shifting the indices of the already-evaluated
points by $\alpha$. This immediately gives a na\"ive algorithm for
computing the sparsest shift of $\fp$: compute $\fp(x+\gamma)$ for
$\gamma=0,1,\ldots,p-1$, and return the $\gamma$ that minimizes the
number of non-zero, non-constant terms. The bit complexity of this
approach is $O(p \log p\ \M(p)\M(\log p))$, which for our applications will
often be less costly than the more sophisticated approaches
of, e.g., \citet{lack-saund-shifts} or \cite{mark-shifts}, precisely
because $p$ will not be very much larger than $\deg \fp$.

\subsection{Overview of Approach}
\label{ssec:overview}

We will make repeated use of the following
fundamental theorem from \citet*{lack-saund-shifts}:

\begin{fact} \label{thm:unique}
Let $\F$ be an arbitrary field and $f \in \F[x]$, and suppose $\alpha \in \F$
is such that $f(x+\alpha)$ has
$t$ non-zero and non-constant terms. 
If $\deg f \geq 2t+1$ then $\alpha$ is the unique
sparsest shift of $f$.
\end{fact}

From this we can see that, if $\alpha$ is the unique sparsest shift of $f$,
then $\alpha_p = \alpha \rem p$ is the unique sparsest shift of $\fp$
\emph{provided that} $\deg \fp \geq 2t+1$. This observation provides the
basis for our algorithm. 

The input to the algorithms will be a modular black box for evaluating a
rational polynomial, as described above, and bounds on the maximal size of
the unknown polynomial. Note that such bounds are a necessity in any type
of black-box interpolation algorithm, 
since otherwise we could never be sure that
the computed polynomial is really equal to the black-box function at
\emph{every} point. Specifically, we require
$B_A, B_T, B_H, B_N \in \NN$ such that
\begin{align*}
\sz{\alpha} &\leq B_A, \\
t &\leq B_T, \\
\sz{c_i} &\leq B_H, \quad \mbox{for}~0 \leq i \leq t,\\
\log_2 n &\leq B_N.
\end{align*}
By considering the following polynomial:
\[c (x-\alpha)^n + (x-\alpha)^{n-1} + \cdots + (x-\alpha)^{n-t+1},\]
we see that these bounds are independent --- that is, none is
polynomially-bounded by the others --- and therefore are all necessary.

We are now ready to present the algorithm for computing the sparsest
shift $\alpha$ almost in its entirety. The only part of the algorithm
left unspecified is an \emph{oracle} which, based on the values of the
bounds, produces primes to use.  
We want primes $p$ such that $\deg \fp \geq 2t+1$, which
allows us to recover one modular image of the sparsest shift $\alpha$.
But since we do not know the exact value of $t$ or the degree $n$ of $f$
over $\QQ[x]$, we define some prime $p$ to 
be a \emph{good prime for sparsest shift computation} if and only if 
$\deg \fp \geq \min\{2B_T+1,n\}$.
For the remainder of this section, ``good prime'' means ``good prime for
sparsest shift computation.''
Our oracle 
indicates when enough primes have been produced so that at least one of
them is guaranteed to have been a good prime, which is necessary for
the procedure to terminate.
The details of how to construct such an oracle will be
considered in \ref{sec:primes}.

\begin{algorithm}{alg:ss}[Computing the sparsest shift]
\item \begin{list}{$\circ$}{\itemsep=0pt}
	\item A modular black box for an unknown polynomial $f \in \QQ[x]$
	\item Bounds $B_A,B_T,B_H,B_N \in \NN$ as described above
	\item An oracle which produces primes and indicates when at
          least one good prime must have been produced
\end{list}
\smallskip
\item A sparsest shift $\alpha$ of $f$.
\begin{block}
\item $P \gets 1$, \qquad $\pset \gets \emptyset$
\begin{whileblock}{$\log_2 P < 2B_A + 1$ \algolabel{termcond}}
	\item $p \gets$ new prime from the oracle \algolabel{choosep}
	\item Evaluate $f(i) \rem p$ for $i=0,1,\ldots,p-1$ \algolabel{bb}
	\item Use dense interpolation to compute $\fp$
		\algolabel{denseinterp}
	\begin{ifblock}{$\deg \fp \geq 2B_T + 1$ \algolabel{degtest}}
      \item Use dense interpolation to compute $\fp(x+\gamma)$ for
         $\gamma=1,2,\ldots,p-1$%
         \algolabel{ssinterp}
      \item $\alpha_p \gets$ the unique sparsest shift of $\fp$
         \label{alg:ss-ap}
		\item $P \gets P \cdot p$,\qquad $\pset \gets \pset \bigcup \{p\}$
	\end{ifblock}
	\begin{elifblock}{$P=1$ and oracle indicates $\geq 1$ good prime
   has been produced
		\algolabel{densetest}}
		\algolabel{densecase}
      \item $q \gets$ least prime such that $\log_2 q > 2B_T B_A + B_H$
         (computed directly)
      \item Evaluate $f(i) \rem q$ for $i=0,1,\ldots,2B_T$
      \item Compute $f\in\QQ[x]$ with 
         $\deg f \leq 2B_T$ by dense interpolation in $\ZZ_q[x]$
         followed by rational reconstruction on the coefficients
         \label{alg:ss:denseinterp}
		\item \RETURN A sparsest shift $\alpha$ 
         computed by a univariate
         algorithm from \citet{mark-shifts} on input $f$
         \label{alg:ss:densess}
	\end{elifblock}
\end{whileblock}
\item \RETURN The unique $\alpha=a/b \in\QQ$ 
such that $|a|,b \leq 2^{B_A}$
and $a \equiv b \alpha_p \bmod p$ for each $p\in\pset$, using Chinese
remaindering and rational reconstruction
\label{alg:ss-reta}
\end{block}
\end{algorithm}

\begin{theorem}
With inputs as specified, \ref{alg:ss} correctly returns a sparsest
shift $\alpha$ of $f$.
\end{theorem}
\begin{proof}
Let $f,B_A,B_T,B_H,B_N$ be the inputs to the algorithm, and
suppose $t,\alpha$ are as specified in \ref{eqn:ssrep}.

First, consider the degenerate case where $n \leq 2B_T$, i.e., the bound
on the sparsity of the sparsest shift is at least half the actual degree
of $f$. Then, since each $\fp$ can have degree at most $n$ (regardless
of the choice of $p$), the condition of Step~\short\ref{alg:ss-degtest}
will never be true. Hence Steps~\short\ref{alg:ss-densecase} will
eventually be executed. The size of coefficients over the standard power
basis is bounded by $2B_TB_A+B_H$ since $\deg f \leq 2B_T$, 
and therefore $f$
will be correctly computed on Step~\short\ref{alg:ss-denseinterp}.
In this case, \ref{thm:unique} may not apply, i.e. the sparsest shift
may not be unique, but the algorithms from \citet{mark-shifts} will
still produce a sparsest shift of $f$.

Now suppose instead that $n \geq 2B_T+1$. The oracle eventually produces
a good prime $p$, so that $\deg \fp \geq 2B_T+1$. Since
$t \leq B_T$ and $\fp$ has at most $t$ non-zero and
non-constant terms in the $(\alpha \rem p)$-shifted power basis,
the value computed as $\alpha_p$ on 
Step~\short\ref{alg:ss-ap} is exactly $\alpha \rem p$, 
by \ref{thm:unique}. The value of $P$ will
also be set to $p > 1$ here, and can only increase. So the condition of
Step~\short\ref{alg:ss-densetest} is never true. Since the numerator and
denominator of $\alpha$ are both bounded above by $2^{B_A}$, we can use
rational reconstruction to compute $\alpha$ once we have the image modulo $P$
for $P \geq 2^{2B_A+1}$. Therefore, when we reach
Step~\short\ref{alg:ss-reta},
we have enough images $\alpha_p$ to recover and return the correct value of
$\alpha$.
\end{proof}

We still need to specify which algorithm to use to compute the sparsest
shift of a densely-represented $f\in\QQ[x]$ on
Step~\short\ref{alg:ss:densess}. To make \ref{alg:ss} completely
deterministic, we should use the univariate symbolic algorithm from
\citet[Section 3.1]{mark-shifts}, although this will have very high
complexity. Using a probabilistic algorithm instead gives the following,
which follows directly from the referenced work.

\begin{theorem}\label{thm:ssdensecost}
   If the ``two projections'' algorithm of \citet[Section 3.3]{mark-shifts} 
   is used
   on Step~\short\ref{alg:ss:densess}, then
   Steps~\short\ref{alg:ss-densecase} of \ref{alg:ss} can be performed
   with $O(B_T^2 \M(B_T^4 B_A + B_T^3 B_H))$ bit operations,
   plus $O(\bb B_T \M(B_T B_A + B_H))$ bit operations for the
   black-box evaluations.
\end{theorem}

The precise complexity analysis proving that the entire
\ref{alg:ss} has
bit complexity polynomial
in the bounds given 
depends heavily on the size and number of primes $p$ that are used, and
so must be postponed until \ref{ssec:sscomplexity},
after our discussion on choosing primes.

\begin{example}
Suppose we are given a modular black box for the following unknown
polynomial:
\begin{align*}
f(x)  =  x^{15} & -
45x^{14}+945x^{13}-12285x^{12}+110565x^{11}-729729x^{10}\\
&+3648645x^9 - 14073345x^8 +42220035x^7-98513415x^6+\\
&177324145x^5 -241805625x^4 +241805475x^3 - 167403375x^2\\
&+71743725x-14348421,
\end{align*}
along with the bounds $B_A=4$, $B_T=2$, $B_H=4$, and $B_N=4$.
One may easily confirm that $f(x) = (x-3)^{15} - 2(x-3)^5$, and hence
these bounds are actually tight.  

Now suppose the oracle produces $p=7$ in
Step~\short\ref{alg:ss-choosep}.  We use the black box to find
$f(0),f(1),\ldots,f(6)$ in $\ZZ_7$, and 
dense interpolation to compute
\[
f^{(7)}(x) = 5x^5+2x^4+3x^3+6x^2+x+4.
\]

Since $\deg f^{(7)} = 5 \geq 2B_T+1$, we move on to
Step~\short\ref{alg:ss-ap} and compute each $f^{(7)}(x+\gamma)$
with $\gamma=1,2,\ldots,6$. Examining these, we see that 
$f^{(7)}(x+3) = 5x^5 + x^3$ 
has the fewest non-zero and non-constant terms, 
and so set $\alpha_7$ to 3 on
Step~\short\ref{alg:ss-ap}. This means
the sparsest shift must be congruent to 3 modulo 7. 
This provides a single modular image for use in Chinese remaindering and
rational reconstruction on Step~\short\ref{alg:ss-reta}, after enough
successful iterations for different primes~$p$.
\end{example}

\subsection{Conditions for Success}

We have seen that, provided $\deg f > 2B_T$, a good prime
$p$ is one such that $\deg \fp > 2B_T$. The following theorem provides
(quite loose) sufficient conditions on $p$ to satisfy this requirement.

\begin{theorem} \label{thm:sspcond}
Let $f \in \QQ[x]$ as in \ref{eqn:ssrep} and $B_T \in \NN$ such that 
$t \leq B_T$.
Then, for some prime $p$,
the degree of $\fp$ is greater than $2B_T$ whenever the
following hold:
\begin{itemize}
\item $c_t \nequiv 0 \mod p$;
\item $\forall i \in \{1,2,\ldots,t-1\},\ e_t \nequiv e_i \mod (p-1)$;
\item $\forall i \in \{1,\ldots,2B_T\},\ e_t \nequiv i \mod (p-1)$.
\end{itemize}
\end{theorem}

\begin{proof}

The first condition guarantees that the last term of $\fp(x)$ as in
\ref{eqn:fp}
does not vanish. We also know
that there is no other term with the same degree from the second condition.
Finally, the third condition tells us that the degree of the last term will be
greater than $2B_T$. Hence the degree of $\fp$ is greater than $2B_T$.
\end{proof}

For purposes of computation it will be convenient to simplify the above
conditions to two non-divisibility requirements, on $p$ and $p-1$
respectively:

\begin{corollary} \label{cor:ssdivcond}
Let $f,B_T,B_H,B_N$ be as in the input to \ref{alg:ss} with
$\deg f > 2B_T$.
Then there exist $C_1,C_2 \in \NN$ with $\log_2 C_1 \leq 2B_H$ and
$\log_2 C_2 \leq B_N(3B_T-1)$ such that $\deg \fp > 2B_T$ whenever
$p \nmid C_1$ and $(p-1) \nmid C_2$.
\end{corollary}

\begin{proof}
Write $f$ as in $\ref{eqn:ssrep}$. We will use the sufficient conditions
given in \ref{thm:sspcond}. Write $|c_t| = a/b$ for $a,b \in \NN$
relatively prime. In order for $c_t \rem p$ to be well-defined and not
zero, neither $a$ nor $b$ can vanish modulo $p$. This is true whenever
$p \nmid ab$. Set $C_1 = ab$.
Since $a,b \leq 2^{B_H}$, $\log_2 C_1 = \log_2 (ab) \leq 2B_H$.

Now write
\[C_2 = \prod_{i=1}^{t-1} (e_t - e_i) \cdot \prod_{i=1}^{2B_T}
(e_t-i).\]
We can see that the second and third conditions of \ref{thm:sspcond} are
satisfied whenever $(p-1) \nmid C_2$. Now, since each integer $e_i$ is
distinct and positive, and $e_t$ is the greatest of these, each
$(e_t-e_i)$ is a positive integer less than $e_t$. Similarly, since 
$e_t = \deg f > 2B_T$, each $(e_t-i)$ in the second product is also a
positive integer less than $e_t$. Therefore, using the fact that 
$t \leq B_T$, we see $C_2 \leq e_t^{3B_T-1}$. Furthermore, $e_t \leq 2^{B_N}$,
so we know that $\log_2 C_2 \leq B_N(3B_T-1)$.
\end{proof}

A similar criteria for success is required in \cite{BlaHar09}, and
they employ Linnik's theorem to obtain a polynomial-time algorithm for
polynomial identity testing. Linnik's theorem was also employed in
\cite{Giesbrecht-Roche:2007} to yield a much more expensive deterministic
polynomial-time algorithm for finding sparse shifts than the one
presented here.

\section{Interpolation} 
\label{sec:interp}

Once we know the value of the sparsest shift $\alpha$ of $f$, we can
trivially construct a modular black box for the $t$-sparse polynomial
$f(x+\alpha)$ using the modular black box for $f$. Therefore, for the
purposes of interpolation, we can assume $\alpha=0$, and focus only on
interpolating a $t$-sparse polynomial $f \in \QQ[x]$ given a modular
black box for its evaluation. The basic techniques of this section are,
for the most part, known in the literature.  However, a unified
presentation in terms of bit complexity for our model of modular
black boxes will be helpful.

For convenience, we restate the notation for $f$ and $\fp$, given a
prime $p$:
\begin{align}
f & = c_0 + c_1 x^{e_1} + c_2 x^{e_2} + \cdots + c_t x^{e_t}, \label{eqn:finterp} \\
\fp &  = (c_0 \rem p) + (c_1 \rem p)x^{e_1 \remo (p-1)} + \cdots
  + (c_t \rem p)x^{e_t \remo (p-1)}. \label{eqn:fpinterp}
\end{align}

Again, we assume that we are given bounds $B_H$, $B_T$, and $B_N$ on
$\max_i \sz{c_i}$, $t$, and $\log_2 \deg f$, respectively.
We also introduce the notation $\tau(f)$, which is defined to be the
number of distinct non-zero, non-constant 
terms in the univariate polynomial $f$.

This algorithm will again use the polynomials $\fp$ for primes $p$, but
now rather than a degree condition, we need $\fp$ to have the maximal
number of non-constant terms.
So we define a prime $p$ to be a \emph{good prime for interpolation}
if and only if $\tau(\fp) = t$. Again, the term ``good prime'' refers to
this kind of prime for the remainder of this section.

Now suppose we have used modular evaluation and dense interpolation
(as in \ref{alg:ss}) to recover the polynomials $\fp$ for $k$ distinct
good primes $p_1,\ldots,p_k$.
We therefore have $k$ images of
each exponent $e_i$ modulo $(p_1-1),\ldots,(p_k-1)$. Write each of these
polynomials as:
\begin{equation}\label{eqn:fpimage}
f^{(p_i)} =  c_0^{(i)} + c_1^{(i)} x^{e_1^{(i)}} + \cdots + c_t^{(i)}
  x^{e_t^{(i)}}.
\end{equation}
Note that it is \emph{not} generally the case that $e_j^{(i)} =
e_j\remo (p_i-1)$. Because we don't know how to associate the
exponents in each polynomial $f^{(p_i)}$ with their pre-image in
$\ZZ$, a simple Chinese remaindering on the exponents will not
work. 
Possible approaches are provided by \cite{Kaltofen:1988}, 
\cite{Kaltofen-Lakshman-Wiley:1990} or \cite{Avendano-Krick-Pacetti:2006}.
However, the most suitable approach for our purposes is the clever technique of 
\cite{Garg-Schost:2009}, based on ideas of \citet{GriKar87}.
We interpolate the polynomial
\begin{equation}\label{eqn:sympoly}
g(z) = (z-e_1)(z-e_2)\cdots(z-e_t),
\end{equation}
whose coefficients are symmetric functions in the $e_i$'s.
Given $f^{(p_i)}$, we have all the values of $e_j^{(i)} \remo (p_i-1)$ for
$j=1,\ldots,t$; we just don't know the order. But since $g$ is not
dependent on the order, we can compute $g \bmod (p_i-1)$ for
$i=1,\ldots,k$, and then find the roots of $g\in\ZZ[x]$ to determine the
exponents $e_1,\ldots,e_t$.  Once we know the exponents, we
recover the coefficients from their images modulo each prime.  The
correct coefficient in each $\fp$ can be identified because the
residues of the exponents modulo $p-1$ are unique, for each chosen
prime $p$.  This approach is made explicit in the following algorithm.

\begin{algorithm}{alg:si}[Sparse Polynomial Interpolation over
  \protect{$\QQ[x]$}]
\item \begin{list}{$\circ$}{\itemsep=0pt}
  \item A modular black box for unknown $f \in \QQ[x]$
  \item Bounds $B_H$ and $B_N$ as described above
  \item An oracle which produces primes and indicates when at least
    one good prime must have been returned
  \end{list}
\medskip
\item $f\in\QQ[x]$ as in \ref{eqn:finterp}

\begin{block}
\item $Q \gets 1$,\quad $P \gets 1$, \quad $k \gets 1$, \quad $t \gets 0$
\begin{whileblock}{$\log_2 P<2B_H+1$ or $\log_2 Q< B_N$ \\ or the oracle does
not guarantee a good prime has been produced
\algolabel{whilep}}
  \item $p_k \gets$ new prime from the oracle \algolabel{choosep}
  \item Compute $f^{(p_k)}$ by black box calls and dense interpolation
    \algolabel{interp}
  \begin{ifblock}{$\tau(f^{(p_k)}) > t$}
    \item $Q \gets p_k-1$, $P \gets p_k$, $t \gets \tau(f^{(p_k)})$,
      $p_1 \gets p_k$, $f^{(p_1)} \gets f^{(p_k)}$, $k \gets 2$
  \end{ifblock}\begin{elifblock}{$\tau(f^{(p_k)}) = t$}
    \item $Q \gets \lcm(Q,p_k-1)$, \quad $P \gets P\cdot p_k$,
    	\quad $k \gets k+1$
  \end{elifblock}
\end{whileblock}

\item \FOR $i\in\{1,\ldots,k-1\}$ \DO
  \begin{block}
  \item $g^{(p_i)}\gets\prod_{1\leq j\leq t} (z-e_j^{(i)})\mod p_i-1$
    \algolabel{computeg}
  \end{block}
\item Construct $g=a_0+a_1z+a_2z^2+\cdots+a_tz^t\in\ZZ[x]$ such that
  $g\equiv g^{(p_i)}\bmod p_i-1$ for $1\leq i< k$, by Chinese
  remaindering \algolabel{makeg}
\item Factor $g$ as $(z-e_1)(z-e_2)\cdots (z-e_t)$ to determine
  $e_1,\ldots,e_t\in\ZZ$\algolabel{factorg}
\item \FOR $1\leq i\leq t$ \DO
\begin{block}
\item \FOR $1\leq j\leq k$ \DO
\begin{block}
\item Find the exponent $e_{\ell_j}^{(j)}$ of $f^{(p_j)}$ 
  such that $e_{\ell_j}^{(j)}\equiv e_i\bmod p_j-1$\algolabel{findexp}
\end{block}
\item Reconstruct $c_i\in\QQ$ by Chinese remaindering from residues
    \hbox to 0pt{$c_{\ell_{1}}^{(1)},\ldots, c_{\ell_{k}}^{(k)}$}
\end{block}
\item Reconstruct $c_0\in\QQ$ by Chinese remaindering from residues
    $c_{0}^{(1)}$, \ldots, $c_{0}^{(k)}$
\end{block}
\end{algorithm}

\noindent
The following theorem follows from the above discussion.

\begin{theorem}
  \ref{alg:si}
  works correctly as stated.
\end{theorem}

Again, this algorithm runs in polynomial time in the bounds given, but
we postpone the detailed complexity analysis until
\ref{ssec:sicomplexity}, after we discuss how to choose primes from the
``oracle''.
Some small practical improvements may be gained
if we use \ref{alg:si} to interpolate $f(x+\alpha)$ after running
\ref{alg:ss} to determine the sparsest shift $\alpha$, since in this
case we will have a few previously-computed polynomials
$f^{(p)}$. However, we do not explicitly consider this savings in our
analysis, as there is not necessarily any asymptotic gain.

Now we just need to analyze the conditions for primes $p$ to be good.
This is quite similar to the analysis of the sparsest shift
algorithm above, so we omit many of the details here.

\begin{theorem}\label{thm:interp-primes}
  Let $f, B_T, B_H, B_N$ be as above. There exist $C_1,C_2 \in \NN$ with
  $\log_2 C_1 \leq 2B_HB_T$ and 
  $\log_2 C_2 \leq \frac{1}{2}B_NB_T(B_T-1)$ such that
  $\tau(\fp)$ is maximal whenever $p\nmid C_1$ and $(p-1)\nmid C_2$.
\end{theorem}
\begin{proof}
  Let $f$ be as in \ref{eqn:finterp}, write $|c_i| = a_i/b_i$ in lowest
  terms for $i=1,\ldots,t$, and define
  \[
  C_1 = \prod_{i=1}^t a_i b_i,\quad
  C_2 = \prod_{i=1}^t \prod_{j=i+1}^t (e_j-e_i).
  \]
  Now suppose $p$ is a prime such that $p \nmid C_1$ and 
  $(p-1) \nmid C_2$. From the first condition, we see that each 
  $c_i \bmod p$ is well-defined and nonzero, and so none of the terms of
  $\fp$ vanish. Furthermore, from the second condition, 
  $e_i \nequiv e_k \bmod p-1$ for all $i \neq j$, so that none of the
  terms of $\fp$ collide. Therefore $\fp$ contains exactly $t$
  non-constant terms. The bounds on $C_1$ and $C_2$ follow from the facts
  that each $a_i,b_i \leq 2^{B_H}$ and each difference of exponents is
  at most $2^{B_N}$.
\end{proof}

\section{Generating primes}
\label{sec:primes}

We now turn our attention to the problem of generating primes for the
sparsest shift and interpolation algorithms.  In previous sections we
assumed we had an ``oracle'' for this, but now we present an explicit
and analyzed algorithm.

The definition of a ``good prime'' is not the same for the algorithms in
\ref{sec:sshift} and \ref{sec:interp}. However,
\ref{cor:ssdivcond} and \ref{thm:interp-primes} provide a unified
presentation of sufficient conditions for primes being ``good''. 
Here we call a prime which satisfies those sufficient conditions a
\emph{useful prime}. So every useful prime is good (with the bounds
appropriately specified for the relevant algorithm), but some good
primes might not be useful.

We first describe a set $\P$ of primes such that the number and
density of useful primes within the set is sufficiently high. 
We
will assume that there exist numbers $C_1,C_2$, and useful primes $p$
are those such that $p\ndivs C_1$ and $(p-1)\ndivs C_2$.  The numbers
$C_1$ and $C_2$ will be unknown, but we will assume we are given
bounds $\beta_1$, $\beta_2$ such that $\log_2 C_1\leq \beta_1$ and $\log_2 C_2\leq
\beta_2$. Suppose we want to find $\ell$ useful primes.  We construct
$\P$ explicitly, of a size guaranteed to contain enough useful primes,
then enumerate it.

The following fact is immediate from \cite{Mikawa:2001}, though it has
been somewhat simplified here, and the use of (unknown) constants is
made more explicit.  This will be important in our computational
methods.

For $q\in\ZZ$, let $S(q)$ be the smallest prime $p$ such that
$q\divs (p-1)$.

\begin{fact}[\citealt{Mikawa:2001}]
  \label{fact:Mikawa} 
  There exists a constant $\mu>0$, such that for all $n>\mu$, and for
  all integers $q\in \{n,\ldots,2n\}$ with fewer than $\mu n/\log^2 n$ 
  exceptions, we have $S(q)<q^{1.89}$.
\end{fact}

Our algorithms for generating useful primes require explicit knowledge
of the value of the constant $\mu$ in order to run correctly.
So we will assume that we know $\mu$ in what follows.  To
get around the fact that we do not, we simply start by assuming that
$\mu=1$, and run any algorithm depending upon it.  If the algorithm
fails we simply double our estimate for $\mu$ and repeat.  At most a
constant number of doublings is required.  We make no claim this is
particularly practical.

For convenience we define
\[
\Upsilon(x)=\frac{3x}{5\log x}-\frac{\mu x}{\log^2x}.
\]

\begin{theorem}
  \label{thm:P}
  Let $\log_2 C_1\leq \beta_1$, $\log_2 C_2\leq \beta_2$ and $\ell$ be as
  above.  Let $n$ be the smallest integer such that $n> 21$, $n> \mu$
  and $\Upsilon(n)>\beta_1+\beta_2+\ell$.  Define
  \[
  \Q =\{q~\mbox{prime}: n\leq q<2n~\mbox{and}~S(q)<q^{1.89} \},\quad
  \P =\{S(q): q\in\Q\}.
  \]
  Then the number of primes in $\P$ is at least
  $\beta_1+\beta_2+\ell$, and the number of useful primes in $\P$, such
  that $p\ndivs C_1$ and $(p-1)\ndivs C_2$, is at least $\ell$.  For
  all $p\in\P$ we have $p\in O((\beta_1+\beta_2+\ell)^{1.89}\cdot
  \log^{1.89}(\beta_1+\beta_2+\ell))$.
\end{theorem}
\begin{proof}
  By \cite{Rosser-Schoenfeld:1962}, the number of primes between $n$
  and $2n$ is at least $3n/(5\log n)$ for $n\geq 21$.  Applying
  \ref{fact:Mikawa}, we see $\#\Q\geq 3n/(5\log n)-\mu n/\log^2 n$
  when $n\geq\max\{\mu,21\}$.  
  Now suppose $S(q_1) = S(q_2)$ for $q_1,q_2\in\Q$. If $q_1<q_2$, then
  $S(q_1) > q_1^2$, a contradiction with the definition of $\Q$. So we
  must have $q_1=q_2$, and hence
  \[
  \#\P=\#Q\geq\Upsilon(n)>\beta_1+\beta_2+\ell.
  \]
  We know that there are at most $\log_2 C_1\leq  \beta_2$ primes
  $p\in\P$ such that $p\divs C_1$.  We also know that there are at
  most $\log_2 C_2\leq \beta_2$ primes $q\in\Q$ such that $q\divs
  C_2$, and hence at most $\log_2 C_2$ primes $p\in\P$ such that
  $p=S(q)$ and $q\divs (p-1)\divs C_1$.  Thus, by construction $\P$
  contains at most $\beta_1+\beta_2$ 
  primes that are not useful out of $\beta_1+\beta_2+\ell$ total
  primes.

  To analyze the size of the primes in $\P$, we note that to make
  $\Upsilon(n)>\beta_1+\beta_2+\ell$, we have
  $n\in\Theta((\beta_1+\beta_2+\ell)\cdot \log(\beta_1+\beta_2+\ell))$ and each $q\in\Q$
  satisfies $q\in O(n)$. Elements of $\P$ will be of magnitude at most
  $(2n)^{1.89}$ and hence $p\in O((\beta_1+\beta_2+\ell)^{1.89}\log^{1.89}(\beta_1+\beta_2+\ell))$.
\end{proof}

Given $\beta_1$, $\beta_2$ and $\ell$ as above (where $\log_2 C_1\leq \beta_1$ and
$\log_2 C_2\leq \beta_2$ for unknown $C_1$ and $C_2$), we generate the primes in
$\P$ as follows.

Start by assuming that $\mu=1$, and compute $n$ as the smallest
integer such that $\Upsilon(n)>\beta_1+\beta_2+\ell$, $n\geq \mu$ and $n\geq
21$.  List all primes between $n$ and $2n$ using a Sieve of
Eratosthenes. For each prime $q$ between $n$ and $2n$, 
determine $S(q)$, if it is less than $q^{1.89}$, by simply checking if
$kq+1$ is prime for $k=1,2,\ldots, \floor{q^{0.89}}$. If we find a
prime $p=S(q)<q^{1.89}$, add $p$ to $\P$. This is repeated until $\P$
contains $\beta_1+\beta_2+\ell$ primes.  If we are unable to find this number
of primes, we have underestimated $\mu$ (since \ref{thm:P} guarantees
their existence), so we double $\mu$ and restart the process.
Obviously in practice we would not redo primality tests already
performed for smaller $\mu$, so really no work need be wasted.

\begin{theorem}\label{thm:determ-primes}
  For $\log_2 C_1\leq \beta_1$, $\log_2 C_2\leq \beta_2$, $\ell$, and $n$ as
  in \ref{thm:P}, we can generate $\beta_1+\beta_2+\ell$ elements of $\P$ with
  $O((\beta_1+\beta_2+\ell)^2\cdot\log^{7+o(1)}(\beta_1+\beta_2+\ell))$ bit
  operations.  At least $\ell$ of the primes in $\P$ will be useful.
\end{theorem}
\begin{proof}
  The method and correctness follows from the above discussion.  The
  Sieve of Eratosthenes can be run with $O(n\log\log\log n)$ bit
  operations (see \cite{Knu81}, Section 4.5.4), and returns $O(n/\log
  n)$ primes $q$ between $n$ and $2n$.  Each primality test of $kq+1$
  can be done with $(\log n)^{6+o(1)}$ bit operations
  \citep{Lenstra-Pomerance:2005}, so the total cost is $O(n^2(\log
  n)^{5+o(1)})$ bit operations.  Since $n\in O((\beta_1+\beta_2+\ell)\cdot
  \log(\beta_1+\beta_2+\ell))$ the stated complexity follows.
\end{proof}

The analysis of our methods will be significantly improved when more
is discovered about the behavior of the least prime congruent to one
modulo a given prime, which we have denoted $S(q)$. An asymptotic
lower bound of $S(q) \in \Omega(q\log^2 q)$ is conjectured in
\citet{GraPom90}, and we have employed the upper bound from
\citet{Mikawa:2001} of $S(q)\in O(q^{1.89})$ (with exceptions). From
our own brief computational search we have evidence that the
conjectured lower bound may well be an upper bound: for all primes
$q\leq 2^{32}$, $S(q)<2q\ln^2 q$. If something similar could be proven
to hold asymptotically (even with some exceptions), the complexity
results of this and the next section would be improved
significantly. In any case, the actual cost of the algorithms
discussed will be a reflection of the true behavior of $S(q)$, even
before it is completely understood by us.

Even more improvements might be possible if this rather complicated
construction is abandoned altogether, as useful primes would naturally
seem to be relatively plentiful.
In particular, one would expect that if we randomly choose
primes $p$ directly from a set which has, say, $4(\beta_1+\beta_2+\ell)$
primes, we might expect that the probability that $p\divs C_1$ or
$(p-1)\divs C_2$ to less than, say, $1/4$.  Proving this directly
appears to be difficult.  Perhaps most germane results to this are
lower bounds on the Carmichael Lambda function (which for the product
of distinct primes $p_1$, \ldots, $p_m$ is the LCM of $p_1-1$, \ldots,
$p_m-1$), which are too weak for our purposes.  See
\cite{Erdos-Pomerance:1991}.

\section{Complexity analysis}
\label{sec:complexity}

We are now ready to give a formal complexity analysis for the algorithms
presented in \ref{sec:sshift} and \ref{sec:interp}.
For all algorithms, the complexity is polynomial in the four bounds
$B_A$, $B_T$, $B_H$, and $B_N$ defined in \ref{ssec:overview},
and since these are each bounded above by $\sz{f}$, our algorithms 
will have polynomial complexity in the size of the output if these bounds
are sufficiently tight.

\subsection{Complexity of Sparsest Shift Computation}
\label{ssec:sscomplexity}

\ref{alg:ss} gives our algorithm to compute the sparsest shift $\alpha$
of an unknown polynomial $f \in \QQ[x]$ given bounds $B_A$, $B_T$,
$B_H$, and $B_N$ and an oracle for choosing primes. The details of this
oracle are given in \ref{sec:primes}.

To choose primes, we set $\ell=2B_A+1$,
and $\beta_1=2B_H$ and $\beta_2=B_N(3B_T-1)$ (according to
\ref{cor:ssdivcond}). For the sake of notational brevity, define 
$\bs =  B_A + B_H + B_NB_T$ so that $\beta_1+\beta_2+\ell \in O(\bs)$.

\begin{theorem} \label{thm:ss-determ-complex}
  Suppose $f\in\QQ[x]$ is an unknown polynomial given by a black box, with
  bounds $B_A$,$B_T$,$B_H$, and $B_N$ given as above.
  If $\deg f > 2B_T$, then
  the sparsest shift $\alpha\in\QQ$ of $f$
  can be computed deterministically using
  \[ O\left( B_A \bs^{1.89} \cdot \log^{2.89} \bs 
     \cdot \M(\bs^{1.89} \log^{1.89} \bs) \cdot \M(\log \bs) \right) \]
  bit operations, plus 
  $O(\bb \bs^{2.89} \log^{1.89}\bs \M(\log \bs))$ 
  bit operations for the black-box evaluations.
\end{theorem}
\begin{proof}
  \ref{alg:ss} will always terminate (by satisfying the conditions of
  Step~\short\ref{alg:ss-termcond})
  after $2B_A+1$ good primes have
  been produced by the oracle. 

  Using the oracle to choose primes,
  and because $\beta_1+\beta_2+\ell \in O(\bs)$,
  $O(\bs^2\log^{7+o(1)}\bs)$
  bit operations are used to compute all the primes on
  Step~\short\ref{alg:ss-choosep}, by \ref{thm:determ-primes}. And
  by \ref{thm:P}, each chosen $p$ is bounded by \linebreak
  $O(\bs^{1.89}\log^{1.89} \bs)$.

  All black-box evaluations are performed on Step~\short\ref{alg:ss-bb};
  there are $p$ evaluations at each iteration, and $O(\bs)$
  iterations, for a total cost of $O(\bb \bs p \cdot \M(\log p))$ bit
  operations. The stated complexity bound follows from the size of each
  prime $p$.

  Steps~\short\ref{alg:ss-densecase} are 
  never executed when $\deg f > 2B_T$. Step~\short\ref{alg:ss-reta} is
  only executed once and never dominates the complexity.

  Dense polynomial interpolation over $\ZZ_p$ is performed at most
  $O(\bs)$ times on Step~\short\ref{alg:ss-denseinterp} and 
  $O(p)$ times at each of $O(B_A)$ iterations through
  Step~\short\ref{alg:ss-ssinterp}. Since $p \gg \bs$, the
  latter step dominates. Using asymptotically fast methods, each
  interpolation of $\fp(x+\gamma)$ uses $O(\M(p)\log p)$ field
  operations in $\ZZ_p$, each of which costs $O(\M(\log p))$ bit
  operations. This gives a total cost over all iterations of
  $O(B_A p \cdot \log p \cdot \M(p) \cdot \M(\log p))$ 
  (a slight abuse of notation
  here since the value of $p$ varies).
  Again, using the fact that $p \in O(\bs^{1.89}\log^{1.89} \bs)$ gives
  the stated result.
\end{proof}

To simplify the discussion somewhat, consider the case that
we have only a \emph{single} bound on the size of
the output polynomial, say $B_f \geq \sz{f}$. By setting
each of $B_T$, $B_H$, and $B_N$ equal to
$B_f$, and by using the multiplication algorithm from \citet{CanKal91},
we obtain the following comprehensive result:

\begin{corollary}
   The sparsest shift $\alpha$ of an unknown polynomial $f\in\QQ[x]$,
   whose shifted-lacunary size is bounded by $B_f$, can be computed
   using
   \[ O\left( B_f^{8.56} \cdot \log^{6.78} B_f \cdot
      (\llog B_f)^2 \cdot \lllog B_f \right) \]
   bit operations, plus
   \[ O\left( \bb B_f^{5.78} \cdot \log^{2.89} B_f 
      \cdot \llog B_f \cdot \lllog B_f\right)\]
   bit operations for the black-box evaluations.
\end{corollary}
\begin{proof}
   The stated complexities follow directly from 
   \ref{thm:ss-determ-complex} above, using the fact that
   $\M(n) \in O(n\log n\llog n)$ and $\bs \in O(B_f^2)$.
   Using the single bound $B_f$, we see that these costs always dominate
   the cost of Steps~\short\ref{alg:ss-densecase} given in
   \ref{thm:ssdensecost}, and so we have the stated general result.
\end{proof}

In fact, if we have no bounds at all \emph{a priori}, we could
start by setting $B_f$ to some small value (perhaps dependent on the
size of the black box or $\bb$), running \ref{alg:ss},
then doubling $B_f$ and running the algorithm again, and so forth
until the same polynomial $f$ is computed in successive
iterations. This can then be tested on random evaluations.  Such an
approach yields an output-sensitive polynomial-time algorithm which
should be correct on most input, though it could certainly be fooled
into early termination.

This is a significant improvement over the algorithms from our original
paper \citep{Giesbrecht-Roche:2007}, which had a dominating factor of
$B_f^{78}$ in the deterministic complexity. 
Also --- and somewhat surprisingly --- our algorithm is competitive
even with the best-known sparsest shift algorithms which require
a (dense) $f\in\QQ[x]$ to be given explicitly as input. By carefully constructing
the modular black box from a given $f\in\QQ[x]$, and being sure to set
$B_T < (\deg f)/2$, we can derive from \ref{alg:ss} a deterministic
sparsest-shift algorithm with bit complexity close to the fastest
algorithms in \citet{mark-shifts}; the dependence on degree $n$ and
sparsity $t$ will be somewhat less, but the dependence on the size of
the coefficients $\log \|f\|$ is greater.

To understand the limits of our computational techniques (as opposed
to our current understanding of the least prime in arithmetic
progressions) we consider the cost of our algorithms
under the optimistic assumption that $S(q)\in O(q\ln^2 q)$, possibly with
a small number of exceptions.  In this case the sparsest shift
$\alpha$ of an unknown polynomial $f\in\QQ[x]$, whose shifted-lacunary
size is bounded by $B_f$, can be computed using
\[ 
O\left( B_f^{5} \cdot \log^{6} B_f \cdot (\llog B_f)^{2} \cdot
  \lllog B_f \right) 
\]
bit operations.  As noted in the previous section, we have verified
computationally that $S(q)\leq 2q\ln^2q$ for $q<2^{32}$.  This would
suggest the above complexity for all sparsest-shift
interpolation problems that we would expect to encounter.

\subsection{Complexity of Interpolation}
\label{ssec:sicomplexity}

The complexity analysis of the sparse interpolation algorithm given in
\ref{alg:si} will be quite similar to that of the sparsest shift
algorithm above. Here, we need $\ell=\max\{2B_H+1,B_N\}$ good primes to
satisfy the conditions of Step~\short\ref{alg:si-whilep}, and from
\ref{thm:interp-primes}, we set $\beta_1=2B_HB_T$ and
$\beta_2=\frac{1}{2}B_NB_T(B_T-1)$. Hence for this subsection we set
$\br=B_T(B_H+B_NB_T)$ so that $\beta_1+\beta_2+\ell\in O(\br)$.

\begin{theorem}
  Suppose $f\in\QQ[x]$ is an unknown polynomial given by a modular black
  box, with bounds $B_T$, $B_H$, $B_N$, and $\br$ given as above.
  The sparse representation of $f$ as in \ref{eqn:ssrep} 
  can be computed with
  \begin{align*}
    O\Big( 
	 & \br\log\br \cdot \M(\br^{1.89}\log^{1.89}\br) \cdot \M(\log \br) \\
	 & + B_N^2 \cdot 
	     \M\big( (B_N+\log B_T) \log(B_N+\log B_T)\big)\Big)
  \end{align*}
  bit operations, plus $O(\bb \br^{2.89} \log^{1.89} \br \M(\log \br))$ 
  bit operations for the black-box evaluations.
\end{theorem}
\begin{proof}
  As in the sparsest-shift computation, the cost of choosing primes in
  Step~\short\ref{alg:si-choosep} requires
  $O(\br^2 \log^{7+o(1)}\br)$ 
  bit operations, and each chosen prime $p_k$ satisfies
  $p_k \in O(\br^{1.89}\log^{1.89}\br)$. The total cost over all
  iterations of Step~\short\ref{alg:si-interp} is also similar to before,
  $O(\br \cdot \M(p_k) \log p_k \cdot \M(\log p_k))$ bit operations,
  plus $O(\bb \br p \M(\log p_k))$ for the black-box evaluations.

  We can compute each $g^{(p_i)}$ in Step~\short\ref{alg:si-computeg}
  using $O(\M(t)\log t)$ ring operations modulo $p_i-1$. Note that
  $k \in O(\ell)$, which is $O(B_H+B_N)$, so the total cost in bit
  operations for all iterations of this step is
  $O((B_H+B_N)\log B_T \cdot \M(B_T)\cdot \M(\log \br))$.

  Step~\short\ref{alg:si-makeg} performs $t$ Chinese Remainderings each
  of $k$ modular images, and the size of each resulting integer is
  bounded by $2^{B_N}$, for a cost of 
  $O(B_T\log B_N\cdot \M(B_N))$ bit operations.

  To factor $g$ in Step~\short\ref{alg:si-factorg}, we can use
  Algorithm~14.17 of \citet{mca}, which has a total cost in bit
  operations of
  \[O\left(B_T^2 \cdot \M(B_N + \log B_T)
    + B_N(B_N+\log B_T) \cdot 
	 \M( (B_N+\log B_T) \log(B_N+\log B_T))\right)\] 
  because the degree of $g$ is $t$,
  $g$ has $t$ distinct roots, and each coefficient is bounded by
  $2^{B_N}$.

  In Step~\short\ref{alg:si-findexp}, we must first compute the modular
  image of $e_i \bmod p_j-1$ and then look through all $t$ exponents
  of $f^{(p_j)}$ to find a match. This is repeated $tk$ times. We can
  use fast modular reduction to compute all the images of each $e_i$
  using $O(\M(B_N)\log B_N)$ bit operations, so the
  total cost is $O(B_T(B_HB_T+B_NB_T+\M(B_N)\log B_N))$
  bit operations.

  Finally, we perform Chinese remaindering and rational reconstruction
  of $t+1$ rational numbers, each of whose size is bounded by $B_H$, for
  a total cost of $O(B_T \cdot \M(B_H) \log B_H)$.

  Therefore we see that the complexity is dominated either by the dense
  interpolation in Step~\short\ref{alg:si-interp} or the root-finding
  algorithm in Step~\short\ref{alg:si-factorg}, depending essentially on
  whether $B_N$ dominates the other bounds.
\end{proof}

Once again, by having only a single bound on the size of the output, the
complexity measures are greatly simplified.

\begin{corollary}
   Given a modular black box for an unknown polynomial 
   $f\in\QQ[x]$ and a bound $B_f$ on the size of its lacunary
   represenation, that representation can be interpolated using
   \[O\left( B_f^{8.67} \log^{4.89}B_f (\llog B_f)^2 \lllog B_f
   \right)\]
   bit operations, plus
   \[O\left( \bb B_f^{8.67} \log^{2.89} B_f \llog B_f \lllog B_f \right)\]
   bit operations for the black-box evaluations.
\end{corollary}

Similar improvements to those discussed at the end of Section
\ref{ssec:sscomplexity} can be obtained under stronger (but unproven)
number theoretic assumptions.

\section{Conclusions and Future Work}
\label{sec:conc}
Here we provide the first algorithm to interpolate an unknown univariate
rational polynomial into the sparsest shifted power basis in time
polynomial in the size of the output. 
The main tool we have introduced is mapping
down modulo small primes where the sparse shift is also mapped
nicely. This technique could be useful for other problems involving
lacunary polynomials as well, although it is not clear how it would
apply in finite domains where there is no notion of ``size''.

There are many further avenues to
consider, the first of which might be multivariate polynomials with a
shift in each variable (see, e.g., \citet{GriLak00}). It would be easy
to adapt our algorithms to this
case provided that the degree in \emph{each variable} is more than twice
the sparsity (this is called a ``very sparse'' shift in
\citet{mark-shifts}). Finding multivariate shifts in the general case
seems more difficult. Even more challenging would be allowing multiple
shifts, for one or more variables --- for example, finding sparse
$g_1,\ldots,g_k\in\QQ[x]$ and 
shifts $\alpha_1,\ldots,\alpha_k\in\QQ$ such that the
unknown polynomial $f(x)$ equals
$g_1(x-\alpha_1)+\cdots+g_k(x-\alpha_k)$.
The most general problem of this type, which we are very far from
solving, might be to compute a minimal-length formula or minimal-size
algebraic circuit for an unknown function. We hope that the small step
taken here might provide some insight towards this ultimate goal.

\section*{Acknowledgement}

The authors would like to thank Igor Shparlinski for pointing out the
paper of \cite{Mikawa:2001}, and for suggesting how to discard
``exceptional'' primes $q$.  This avoids the use of Linnik's theorem,
as employed in \cite{Giesbrecht-Roche:2007}, and improves the
complexity considerably.

The authors would also like to thank Erich Kaltofen for discussions
and sharing of his early unpublished work on rational interpolation,
and \'Eric Schost for discussions and sharing a pre-print of
\cite{Garg-Schost:2009}.

Finally, the authors would like to thank the anonymous reviewers for
their careful readings and useful suggestions.

An extended abstract of a preliminary version of this work appeared at
the MACIS 2007 conference \citep{Giesbrecht-Roche:2007}.


\newcommand{\Gathen}{\relax}

\end{document}